\documentclass[%
reprint
]{revtex4-2}

\usepackage{graphicx}
\usepackage{dcolumn}
\usepackage{bm}

\usepackage{amsmath,amssymb,amsthm,easybmat,verbatim}
\usepackage{enumerate}
\usepackage{hyperref}
\usepackage{color}
\usepackage{bm}
\usepackage{graphicx}
\usepackage{longtable}
\usepackage{array}
\usepackage{tikz}
\usepackage{float}
\usepackage{multirow}
\usepackage{blindtext}

\newtheorem{corollary}{Corollary}

\newtheorem{lemma}{Lemma}

\newtheorem{theorem}{Theorem}

\usepackage{physics}
\usepackage{mathtools}

\begin{document}


\title{Some Entanglement Survives Most Measurements}

\author{Alvin Gonzales$^{1,2}$}
\email{agonza@siu.edu}

\author{Daniel Dilley$^3$}

\author{Mark Byrd$^{2,3}$}
\affiliation{$^1$Intelligence Community Postdoctoral Research Fellowship Program, Argonne National Laboratory, 9700 S. Cass Avenue, Lemont, Illinois, 60439, USA}
\affiliation{$^2$School of Computing, Southern Illinois University Carbondale, Carbondale, Illinois, 62901, USA}
\affiliation{$^3$School of Physics and Applied Physics, Southern Illinois University Carbondale, Carbondale, Illinois, 62901, USA.}
\date{\today}

\begin{abstract}

To prepare quantum states and extract information, it is often assumed that one can perform a perfectly projective measurement. Such measurements can achieve an uncorrelated system and environment state. However, perfectly projective measurements can be difficult or impossible to perform in practice. We investigate the limitations of repeated non-projective measurements in preparing a quantum system. For an $n$-qubit system initially entangled with its environment and subsequently prepared with measurements, using a sequence of weak measurements, we show that some entanglement remains unless one of the measurement operators becomes perfectly projective through an extreme limiting process. Removing initial (unentangled) correlations between a system and its environment and the scenario where measurement outcomes are not tracked are also discussed. We present results for $n$-qubit and $n$-dimensional input states.

\end{abstract}

\maketitle

{\it Introduction}.--
Quantum information processing theory relies on the supposition of perfectly projective measurements in many instances.  They are used for state preparation, error correction and generally for defining completely positive trace-preserving maps \cite{Nielsen_Chuang_Textbook_2011}.  However, they have been under scrutiny recently for the thermodynamics's costs of the implementation \cite{Guryanova_2020IdealProjMeasHaveInfiniteResourceCost,Guryanova_2020IdealProjMeasHaveInfiniteResourceCost,Masanes_2017AGenDerivAndQuantOfThe3rdLawOfThermo, Taranto_2021LandauerVSNernstWhatIsTheTrueCostOfCoolingAQuantSys}.  This begs the questions, what can be done in practice?  What are the implications of restriction to measurements that are not perfectly projective?

For state preparation and error correction in quantum computing, the threshold theorem states that we do not require perfect performance from any component to find the correct results with very high probability \cite{Aharaonov/Ben-Or:Threshold,Kitaev1997, knill_1996ThresholdAccurForQuantComp}.  However, much of our theory relies on our assumed perfection of measurements and with the recent thermodynamic concerns, it is important to try to understand the practicality and implications.  For example, initial system-environment correlations cause complications when defining a map for the system evolution \cite{Pechukas_1994NCP, Pechukas_1995, Alicki_1995, Jordan_Shaji_Sudarshan_2004_DynOfInitialEntangledOpenQS, Shaji_Sudarshan_2005_whosAfraidOfNCPMaps} and can result in temporally correlated errors and/or non-Markovian evolutions that can imply restrictions on error correction methods \cite{Milz_2019CPDivisDoesNotMeanMarkov, Pollock_2018NonMarkovQuantProcCompFrameAndEffCharac, Huo_2021SelfConstTomogOfTemporCorrErr,  shabani_2009MapsForGeneralOQSAndTheoryOfLinQEC, Gonzales_2020SuffCondAndConstForReversGenQuantErrs, Gonzales_2020_ErratumSuffCondsAndConstForReverGenQuantErrs}. So, does a non-perfectly projective measurement leave a state entangled, or otherwise correlated with other states?  In particular, can we decouple qubits from each other and/or from the environment?  

Here, these questions are addressed using non-perfect measurements that are modeled as weak measurements and the strength is increased until the measurement is “almost perfect.”  
Interesting scenarios arise when measurement outcomes are not recorded. These cases are discussed and distinguished here so that some practical modeling and the implications of imperfectly measured states can be explored in a rather general way.  This, of course, has implications for the control of quantum systems and their utility as well as broader implications of simply having to sometimes live with some remaining correlations in particular circumstances.  

This investigation begins by defining and discussing the basic properties of entanglement non-annihilating trace preserving (ENATP) maps. ENATP maps are trace preserving maps that cannot bring to zero the entanglement between the two subspaces of \textit{any} bipartite input state. For convenience, we generally label the two subspaces of the state as the system and environment. ENATP maps were previously investigated in \cite{Lami_2016EntanglementSavingChannels}. In contrast to \cite{Lami_2016EntanglementSavingChannels}, the ENATP maps considered here can be nonlinear due to the renormalization factor for recorded measurement outcomes \cite{Filipov_2021TraceDecreasingMaps}. Next, we consider the effects of non-projective measurements on the initial entanglement of bipartite states.

We identify a disentangling criteria and prove that breaking entanglement via non projective measurements is difficult to attain. Our first result concerns $n$-dimensional input states. It is shown that any finite sequence of full rank measurements with recorded outcomes form an ENATP map. A similar result applies for $n$-qubit input states where the local measurements are full rank.
Thus, a finite sequence of weak measurements \cite{Oreshkov_2005WeakMeasAreUniversal, Brun_2002ASimpleModelOfQuantTraj} with recorded outcomes also form an ENATP map. Weak measurements are general since they can generate any abitrary measurement \cite{Oreshkov_2005WeakMeasAreUniversal}. Consequently, breaking entanglement between an $n$-qubit system and its environment via single qubit measurements with recorded outcomes requires a perfect projective operation.

Our analysis is extended to annihilating initial correlations. We define correlation non-annihilating trace preserving (CNATP) maps as trace preserving maps that cannot bring to zero the correlations between the two subspaces of \textit{any} bipartite input state. Note that correlations in a separable system and environment state can also cause difficulties in describing the evolution of the system \cite{Rodriguez-Rosario_2008_CPMapsAndClassicalCorr}.

Next, when measurement outcomes are not tracked, our ignorance restricts the density matrix that we can assign to the post measurement state to be the convex combination of all the possible outcomes. For an initially entangled state, this convex combination can be unentangled. However, this is a result of our ignorance because all of the possible outcomes are entangled. Multiple examples are provided throughout the text. We conclude with a discussion.

{\it Background}.--\label{sec:background}
We first provide useful definitions and the relation between ENATP maps and entanglement breaking (EB) channels \cite{Holevo_1998QuantumCodingTheorems, Horodecki_2003EntanglementBreakingChannels, Ruskai_2003QubitEntBreakingChannels, De_Pasquale_2012QauntTheNoiseOfAQuantChannelByNoiseAdd, De_Pasquale_2012Erratum:QuantTheNoiseOfAQuantChannByNoiseAdd, Christandl_2019WhenDoComposedMapsBecomeEB}. Note that in these discussions the entanglement considered is partitioned on a bipartite state and it is convenient to think of it as the entanglement between the system and environment. An entanglement breaking channel \cite{Holevo_1998QuantumCodingTheorems, Horodecki_2003EntanglementBreakingChannels, Ruskai_2003QubitEntBreakingChannels, De_Pasquale_2012QauntTheNoiseOfAQuantChannelByNoiseAdd, De_Pasquale_2012Erratum:QuantTheNoiseOfAQuantChannByNoiseAdd, Christandl_2019WhenDoComposedMapsBecomeEB} acts on one half of the system and environment state and breaks the entanglement for \textit{any} input state.  

A closely related concept is that of entanglement-annihilating channels, which considers maps acting on both the system and environment \cite{Filippov_2012LocalTwo-qubitEntAnnihilatingChannels, Filippov_2013BipartiteEntAnnihilatingMapsIFFConds, Filippov_2021TensorProductsOfQuantumMaps}.
A subset of EB channels are called coherence-breaking channels \cite{Bu_2016CoherenceBreakingChannelsAndCoherenceSuddenDeath}, which are channels that break coherence for all input states. Some other similar types of channels are partially entanglement-breaking channels, which are categorized based on the Schmidt rank of the output states \cite{Chruscinski_2006OnPartialEntangleBreakingChannels}, and steerability breaking channels, which are channels that break the steerability of any input state \cite{ku_2022QuantifyingQuantOfChannelsWithoutEntanglement}. In \cite{xu_2022WhatHappensToMultiEntIfOnePartBecomesClass}, the effect of an entanglement-breaking channel on the entanglement of a tripartite state is analyzed.

ENATP and CNATP maps are in some sense opposite of entanglement breaking maps. ENATP (CNATP) maps do not break the entanglement (correlations) between  \textit{any} system and environment state. Note that an $n$-qubit ENATP (CNATP) map is a ENATP (CNATP)
map over $n$-qubit input states.

For this work, the definition of a weak measurement from  \cite{Oreshkov_2005WeakMeasAreUniversal} is used.  A measurement is weak if \textit{all} the outcomes result in very little change to the original state.  The general form of the weak measurement operator is taken to be 
    \begin{align}\label{eq:generalWeak}
        M =q(\mathbb{I}+\hat{\epsilon}),
    \end{align}
where $0\leq q \leq 1$ and $\hat{\epsilon}$ has small norm $\|\hat{\epsilon}\|\ll 1$ \cite{Oreshkov_2005WeakMeasAreUniversal, Brun_2002ASimpleModelOfQuantTraj}. It is known that weak measurements can model any measurement to arbitrary precision \cite{Bennett_1999QuantumNonlocalityWithoutEntanglement, Oreshkov_2005WeakMeasAreUniversal}. Weak measurements have previously been used to detect initial entanglement in quantum batteries \cite{Imai_2023WorkFluctAndEntInQuantBatteries}.

For the entanglement measure we utilize Wootters' concurrence \cite{Wooters_1998EntOfFormationOfAnArbitraryStateOfTwoQubits}
\begin{align}
    C(\rho)=\text{max}\{0,\sqrt{\lambda_1}-\sqrt{\lambda_2}-\sqrt{\lambda_3}-\sqrt{\lambda_4}\},
\end{align}
where the $\lambda_i$ are the eigenvalues in decreasing order of the matrix $\zeta=\rho\tilde\rho$ and 
\begin{align}
    \tilde\rho=(\sigma_y\otimes\sigma_y)\rho^*(\sigma_y\otimes\sigma_y).
\end{align}

We also rely on entanglement witnesses, which are Hermitian operators $W$ that have non negative expectation values for separable states, but has a negative expectation value for at least one entangled state \cite{TERHAL_2000BellIneqAndTheSepCrit, Lewenstein_2000OptimizationOfEntWitness, sperling_2011TheSchmidtNumAsAUnivEntMeas}.

{\it Results}.-- In the following results we assume that the system and environment partitioning is constant. The entanglement and correlations discussed are across the system and environment. The first result concerns the composition of maps, which we use extensively throughout the paper.
\begin{lemma}\label{lemma:compositionsENATP}
    The composition $\mathcal{E}_1\circ\mathcal{E}_2\circ\cdots\circ\mathcal{E}_n$ consisting of $n$ ENATP (CNATP) maps $\mathcal{E}_i$, also forms an ENATP (CNATP) map.
\end{lemma}
\begin{proof}
    The individual maps can be considered to be performed sequentially and hence the composition is also ENATP. From the same arguments, the same results hold for CNATP maps.
\end{proof}
When restricted only to ENATP (CNATP) maps, an infinite sequence of ENATP (CNATP) maps is necessary to bring the entanglement (correlations) to zero. A consequence of Lemma \ref{lemma:compositionsENATP} is that compositions of ENATP maps which act non trivially on different spaces are also ENATP maps.

The following theorem shows the difficulty of breaking entanglement with measurements in experiments. 

\begin{theorem}\label{thm:ddimensionalENATP}
    Let the bipartite initial state consist of an $m$ dimensional system, $n-m$ dimensional environment, and be entangled. Some entanglement remains across the system and environment after an arbitrary finite number of local measurements with known outcomes provided that the local measurement operators are full rank.
\end{theorem}
\begin{proof}
    Let $\rho$ be the initial entangled state and $\rho'=(M_S\otimes M_E)\rho (M_S^\dagger\otimes M_E^\dagger)$ be the unormalized post measurement state. Following a similar proof in \cite{sperling_2011TheSchmidtNumAsAUnivEntMeas}, we can use a witness $W$ to show that $\rho'$ is entangled.

    For any separable state $\sigma$, $\tr(\sigma W)\geq0$. Since $\sigma'=M_S^{-1}\otimes M_E^{-1}\sigma (M_S^{-1}\otimes M_E^{-1})^\dagger$ is an unormalized separable state,
    \begin{align}
        &\tr(\sigma'W)\geq0\\
        &\notag\rightarrow\\
        &\tr(\sigma [(M_S^\dagger)^{-1}\otimes (M_E^\dagger)^{-1}]W(M_S^{-1}\otimes M_E^{-1}))\geq 0
    \end{align}
    and 
    \begin{align}
    W'=[(M_S^\dagger)^{-1}\otimes (M_E^\dagger)^{-1}]W(M_S^{-1}\otimes M_E^{-1})
    \end{align}
    is also a witness (the renormalization factor only scales by a positive factor). Let $W$ be a witness for $\rho$ and thus $\tr(\rho W)<0$. Then we have $\tr(\rho' W')=\tr(\rho W)<0$. Thus, the measurement forms an ENATP map and the result of multiple rounds of measurements follows from Lemma \ref{lemma:compositionsENATP}.
\end{proof}

Theorem \ref{thm:ddimensionalENATP} leads to some immediate consequences for states comprised of qubits.
\begin{theorem}\label{thm:nQubitMeasENATP}[Finite sequence of measurement outcomes with non-zero determinant single qubit operators form an ENATP map]
    Let the bipartite initial state be comprised of an $n$-qubit system, arbitrary environment, and entangled. Some entanglement remains across the system and environment after an arbitrary finite number of rounds of local single qubit measurements on the system with known outcomes provided that the measurement operators are full rank (i.e., rank 2).
\end{theorem}
\begin{proof}
    This result follows from Theorem \ref{thm:ddimensionalENATP} by setting the system state to $n$ qubits and the measurement operators on the system to full rank single-qubit measurements.
\end{proof}

An example of Theorem \ref{thm:nQubitMeasENATP} is the single-qubit measurement
\begin{align}\label{eq:nearlyProjMeas1}
    M_0&=\sqrt{\epsilon}P_-+\sqrt{1-\epsilon}P_+\\
    \label{eq:nearlyProjMeas2}
    M_1&=\sqrt{1-\epsilon}P_-+\sqrt{\epsilon}P_+,
\end{align}
where $0<\epsilon\ll1$ and $P_-$ and $P_+$ are orthogonal projectors. $M_0$ and $M_1$ are asymptotically projective in the sense that they can be arbitrarily close to $P_+$ or $P_-$. Any finite number of rounds with this measurement and known outcomes form an ENATP map, namely, for \textit{any} $n$-qubit initially entangled system and environment input state the post measurement state is entangled. By setting $P_+=\op{0}$, this is a measurement that we might naively assume can prepare a ground state on a quantum computer.

A consequence of Theorem \ref{thm:nQubitMeasENATP} is that a finite sequence of general weak measurements with known outcomes forms an ENATP map.

\begin{corollary}\label{coro:weakMeasKnownOutCan'tBreakEnt}[Finite sequence of weak measurements with recorded outcomes form an ENATP map]
The general form of the weak measurement operator $M$ is given by Eq.~\eqref{eq:generalWeak}. We perform a finite sequence of local weak measurements on the $n$-qubit system and arbitrary environment. If we record the measurement outcomes then the weak measurements form an ENATP map and we cannot annihilate the entanglement of the bipartite input state.

\end{corollary}
\begin{proof}
An arbitrary single qubit weak measurement operator has to be rank 2. Otherwise, it will violate the definition of a weak measurement \cite{Oreshkov_2005WeakMeasAreUniversal} that \textit{all} measurement outcomes result in very little change to all input states. Thus, from Theorem \ref{thm:nQubitMeasENATP}, the result follows.

\end{proof}


The determinant of a sequence of weak measurement operators is equal to the products of the determinants of the individual operators. Each weak measurement operator is rank 2 and has nonzero determinant.
Thus, the determinant is zero and hence from Theorem \ref{thm:nQubitMeasENATP} entanglement is broken only in the limit that the number of weak measurements  goes to infinity. Therefore, a perfect projective operation is required (a rank 1 operation is the outer product of two vectors and is equivalent to an unormalized projection followed by a unitary).

We can also obtain the level of entanglement remaining for a two-qubit input state. For any single local measurement outcome $M$, we have the important property
\begin{align}
\hat{\sigma}_y M^T \hat{\sigma}_y M = \text{Det}(M) \cdot \mathbb{I}.
\end{align}
Let $\rho_R=(M\otimes \mathbb{I})\rho(M^\dagger\otimes \mathbb{I})$ be the unormalized post-measurement state. The square root of the eigenvalues for the calculation of Wootters' concurrence are
\begin{align}
    &\sqrt{\text{Eig}(\rho_R \hat{\sigma}_y \otimes \hat{\sigma}_y \rho_R^* \hat{\sigma}_y \otimes \hat{\sigma}_y)} \nonumber\\
    = &\sqrt{\text{Eig}[(M \otimes \mathbb{I}) \rho (M^\dagger \otimes \mathbb{I})(\hat{\sigma}_y M^* \otimes \hat{\sigma}_y) \rho^* (M^T \hat{\sigma}_y \otimes \hat{\sigma}_y)]} \nonumber\\
    = &\sqrt{\text{Eig}(|\text{Det}(M)|^2 \rho \tilde{\rho})}   \label{eq:eigenPostConcur},
\end{align}
and if we want the exact concurrence, we simply divide by the renormalization factor $\tr(\rho_R) = \tr (M \otimes \mathbb{I} \rho M^\dagger \otimes \mathbb{I})$. Using Eq.~\eqref{eq:eigenPostConcur} we  obtain
\begin{align}
    C(\rho_R/\tr(\rho_R)) = \dfrac{|\text{Det}(M)|}{\tr (M \otimes \mathbb{I} \rho M^\dagger \otimes \mathbb{I})}  C(\rho)
\end{align}
for the post measurement state.

Next, we present results on breaking initial correlations with weak measurements. Note that to ``break initial correlations" we mean that the final state is a product state, $\rho_S\otimes \rho_E$, and that  the correlation matrix $\mathcal{T}_{ij} = \tr (\sigma_i \otimes \sigma_j \rho)$ can be nonzero for a product state. Consider the following form of what we define as ``special" weak measurements. Let
\begin{align}
\label{Eq:special_weak_measurement}
    M_{\pm}(\epsilon,\hat{n}) = \dfrac{1}{2} \big( \varepsilon_+ \mathbb{I} \pm \varepsilon_- \hat{n} \cdot \vec{\sigma} \big),
\end{align}
where
\begin{align}
    \varepsilon_{\pm} = \sqrt{\dfrac{1+\epsilon}{2}} \pm \sqrt{\dfrac{1-\epsilon}{2}},
\end{align}
$\hat{n} = \{ n_x, n_y, n_z \}$ is real, $\hat{n} \cdot \hat{n} = 1$, and $-1\leq\epsilon\leq 1$. Eq.~\eqref{Eq:special_weak_measurement} is a weak measurement when $\abs{\epsilon}$ is small, which leaves the state almost unchanged. Note that $M_+$ and $M_-$ commute.

\begin{theorem}[A special weak measurement with recorded outcome forms a CNATP map]\label{thm:iffForZeroChi}
    For a correlated two-qubit input state $\rho_{SE}$, we measure the system with $M_\pm$. Let $\mathcal{T}'$, $\vec{a}'$, and $\vec{b}'$ be the final correlation matrix, local Bloch vector for the system, and local Bloch vector for the environment, respectively. For a known measurement outcome, the equality $\mathcal{T}' = \vec{a}' \vec{b}'^T$, in other words the final state is a product state, holds iff $\epsilon = \pm1$.
\end{theorem}

An immediate implication of Theorem \ref{thm:iffForZeroChi} is that a finite sequence of special weak measurements with recorded outcomes form two-qubit CNATP maps. This is stated formally as follows.
\begin{corollary}[Finite sequence of special weak measurements with recorded outcomes form a two-qubit CNATP map]
We perform a finite sequence of special local weak measurements on both the system and environment of two-qubit system and environment states. If we know the measurement outcomes then the weak measurements form a two-qubit CNATP map and we cannot annihilate the correlations between any correlated two-qubit input state.
\end{corollary}
\begin{proof}
The sequence of measurement outcomes can be applied sequentially. Then the result follows from Lemma \ref{lemma:compositionsENATP} and Theorem \ref{thm:iffForZeroChi}.
\end{proof}

Finally, we present results for breaking entanglement with non projective measurements when outcomes are not recorded. This scenario can be interpreted as one where the experimenter performs a series of measurements, but forgets the results. The final density matrix that we assign is a convex combination of all the possible outcomes. From our above results, each possible measurement outcome is entangled and the final state is in one of these states. However, the form that we can assign to the post measurement state is limited by our ignorance. The entanglement of this assigned state is the entanglement discussed in the following results. Even in this scenario, the entanglement of the assigned state is difficult to break with non projective measurements. For an initial two-qubit pure state, breaking the initial entanglement with special weak measurements and unknown outcomes requires infinite rounds of measurements. It is convenient to start with a simple case.

\begin{lemma}\label{lemma:weakBrunPure}
Let the two-qubit initial state be a pure state. Consider the weak measurement
\begin{align}\label{eq:weakMeasBrun1}
    M_0&=\sqrt{\dfrac{1+\epsilon}{2}}\op{0}{0}+\sqrt{\dfrac{1-\epsilon}{2}}\op{1}{1}\\
    \label{eq:weakMeasBrun2}
    M_1&=\sqrt{\dfrac{1-\epsilon}{2}}\op{0}{0}+\sqrt{\dfrac{1+\epsilon}{2}}\op{1}{1},
\end{align}
where $0\leq\epsilon< 1$. 

We perform multiple rounds of measurements on the first qubit. If the initial concurrence $C(\ket{\psi})\neq 0$, the concurrence of the post measurement state is non zero for any arbitrary finite number of rounds of the same repeated weak measurement with unknown outcomes.
\end{lemma}

Then using the fact that any special weak measurement Eq.~\eqref{Eq:special_weak_measurement} can be rotated to the forms of Eqs.~\eqref{eq:weakMeasBrun1} and \eqref{eq:weakMeasBrun2} by a local unitary (e.g., by rotating $\hat{n}$ to the z direction) we arrive at the following theorem.
\begin{theorem}[A finite sequence of a special weak measurement with unknown outcomes cannot break entanglement of a pure state]\label{thm:specialWeakMeasEntOneSys}

Let the two-qubit initial state be a pure state. We repeatedly measure the first qubit with a special weak measurement. If the initial concurrence $C(\ket{\psi})\neq 0$, the concurrence of the post measurement state is non zero for any arbitrary finite number of rounds of the same repeated weak measurement with unknown outcomes.
\end{theorem}

We can generalize Theorem \ref{thm:specialWeakMeasEntOneSys} to the case where we measure both systems.
\begin{corollary}[A finite sequence of local special weak measurement with unknown outcomes on the system and environment cannot break entanglement of a pure state]\label{coro:specialWeakUknownBothSys}

Let the two-qubit initial state be a pure state. We measure the first qubit $n$ rounds and the second qubit $y$ rounds with possibly different special weak measurements (the measurement on each qubit stays the same). If the initial concurrence $C(\ket{\psi})\neq 0$, the concurrence of the post measurement state is non zero for any arbitrary finite number of rounds of the weak measurements with unknown outcomes.
\end{corollary}

The following are examples of breaking the entanglement of the assigned state with rank 2 measurements and unknown outcomes. For these examples, the initial entanglement is broken with one round of measurements on both the system and environment and unknown outcomes.

\textbf{Example 1:} Let the initial state be
\begin{align}
    \rho_{SE}=\begin{bmatrix}
    0.5 &0 &0 &a\\
    0 &0 &0 &0\\
    0 &0 &0 &0\\
    a &0 &0 &0.5
    \end{bmatrix},
\end{align}
where $a$ is a parameter that lets us vary the initial entanglement. Let $a=0.002$ and the special weak measurement have the parameters $\epsilon=0.1$ and $\vec{n}=(1,0,0)$.

\textbf{Example 2:} Apply the measurement
\begin{align}
    M_{\pm} = \dfrac{1}{\sqrt{10}}
    \left(
    \begin{array}{cc}
        2 & \pm 1 \\
        \pm 1 & 2
    \end{array}
    \right)
\end{align}
to the state $\dfrac{1}{8} \left( 5 \op{\Phi^+}{\Phi^+} + 3 \op{\Phi^-}{\Phi^-} \right)$, where the Bell state $\ket{\Phi^\pm} = (1/\sqrt{2}(\ket{00}\pm\ket{11}))$.

\textbf{Example 3:} Lastly, there exists initial pure states whose entanglement can be broken by a general weak measurement with unknown outcomes. Consider the initial pure state
\begin{align}
    \ket{\psi} = \text{cos}(\theta/2)\ket{00} + \text{sin}(\theta/2)\ket{11}.
\end{align}
and the weak measurement
\begin{align}
    M_{\pm} = \sqrt{\dfrac{1\mp\epsilon}{2}} \left( \mathbb{I} \pm \epsilon(\mathbb{I}+\hat{\sigma}_x) \right)
\end{align}
for small $\epsilon$ values. Let the state be weakly entangled by setting $\theta=0.0001$. For $\epsilon=0.01$, the general weak measurement operators have $\norm{\hat\epsilon}=\norm{\epsilon(\mathbb{I}+\hat\sigma_x)}\approx 0.02.$

{\it Conclusions}.--
Recent results on the thermodynamic costs of the implementation of perfect projective measurements bring into question their practicality \cite{Guryanova_2020IdealProjMeasHaveInfiniteResourceCost,Guryanova_2020IdealProjMeasHaveInfiniteResourceCost,Masanes_2017AGenDerivAndQuantOfThe3rdLawOfThermo, Taranto_2021LandauerVSNernstWhatIsTheTrueCostOfCoolingAQuantSys}. When measurements are restricted to be non-projective, the ability to break entanglement or correlations via local operations becomes very restricted.

Our results show that disentangling any $n$-qubit bipartite state via single qubit measurements with recorded outcomes, requires either an infinite sequence of rank 2 measurements or performing a rank 1 operation (an unormalized projective measurement operator followed by a unitary). We proved this through a sequence of weak measurements. Removing correlations is also difficult. Special weak measurements with recorded outcomes cannot remove initial correlations.

Let us clarify the assumptions regarding the measurement operator. Since we are assuming that the experimenter cannot design a perfect measurement, we also assume that they do not know the exact form of the measurement operators. We can narrow down our scenario to two outcome measurements because a measurement with more outcomes can be implemented using a sequence of two outcome measurements. For a two-outcome measurement the operators are
\begin{align}
    &M'_0=M_0+E_0\\
    &M'_1=M_1+E_1,
\end{align}
where $M_{0,1}$ are the ideal measurement operators and $E_{0,1}$ are unknown imperfections of the measurement operators. We make the assumption that the errors $E_{0,1}$ prevent $M'_{0,1}$ from being projective \cite{Guryanova_2020IdealProjMeasHaveInfiniteResourceCost}, and thus in the case of single qubit measurements $M_{0,1}'$ are full rank. 

Assume that we are dealing with qubits and we record the zero outcome. The post measurement state for an initial state $\rho$ is $\rho'=M_0'\rho M_0'^\dagger/\tr(M_0'\rho M_0'^\dagger)$. From Theorem \ref{thm:nQubitMeasENATP}, we know that $\rho'$ must be entangled. However, ignorance of $E_{0,1}$ prevents us from knowing the exact form of $\rho'$. Given our limited knowledge, one state we can assign is $\rho'\approx M_0\rho M_0^\dagger/\tr(M_0\rho M_0^\dagger)$, but we know that this is only an approximation of the true state.

Similarly, if measurement results are unknown (we do not know if $M_0'$ or $M_1'$ occurred), our ignorance limits the density matrix that we can assign after measurement despite the state being entangled. Breaking the entanglement of the density matrix that we can assign is also difficult.

Clearly there are many important implications for these results, including, but not limited to, preparation of states for quantum information processing and quantum error correction. For instance, the lack of a ground state starting point degrades the starting fidelity of a quantum computer \cite{Buffoni_2022ThirdLawOfThermoAndtheScalingOfQuantComp}.  For a fully functioning quantum computing device, with error rates well below the threshold for the code, this persistent entanglement may be less of an issue.  However, for noisy intermediate-scale quantum (NISQ) devices the remaining correlations and non-Markovian behavior, will be much more troublesome.  In   addition, the inability to remove entanglement could cause noisy operators to become nonlocal and propagate unwanted correlations throughout the system.
With further investigation, such effects might be overcome by new methods for error mitigation.

{\it Acknowledgements}.--
This research was supported in part by an appointment to the Intelligence Community Postdoctoral Research Fellowship Program at Argonne National Laboratory, administered by Oak Ridge Institute for Science and Education through an interagency agreement between the U.S. Department of Energy and the Office of the Director of National Intelligence. Funding for this research was provided in part by the NSF, MPS under award number PHYS-1820870.

\appendix
\section{Proof of Theorem \ref{thm:iffForZeroChi}}
We first present the detailed proof for the result that special weak measurements do not break initial correlations. Note that to ``break initial correlations" we mean that the final state is a product state, $\rho_S\otimes \rho_E$. We note here that the correlation matrix $\mathcal{T}$ can be nonzero for a product state. Recall the following definition of what we call ``special" weak measurements. Let
\begin{align}
\label{Eq:special_weak_measurement_apdx}
    M_{\pm}(\epsilon,\hat{n}) = \dfrac{1}{2} \big( \varepsilon_+ \mathbb{I} \pm \varepsilon_- \hat{n} \cdot \vec{\sigma} \big),
\end{align}
where
\begin{align}
    \varepsilon_{\pm} = \sqrt{\dfrac{1+\epsilon}{2}} \pm \sqrt{\dfrac{1-\epsilon}{2}},
\end{align}
$\hat{n} = \{ n_x, n_y, n_z \}$ is real, $\hat{n} \cdot \hat{n} = 1$, and $-1\leq\epsilon\leq 1$. Eq.~\eqref{Eq:special_weak_measurement_apdx} is a special weak measurement when $\epsilon$ is small, which leaves the state almost unchanged. Note that $M_+$ and $M_-$ commute.

\begin{proof}
First, we consider outcome $M_i$. We diagonalize the correlation matrix of $\rho^{SE}$ using the unitary $U \otimes V$ \cite{Dilley_2022SO(3)toSu(2)}.  The local unitaries can diagonalize the correlation matrix because they induce left and right rotations on the correlation matrix \cite{Dilley_2022SO(3)toSu(2)}. We see that
\begin{align}
    &(M_iU^\dagger \otimes V^\dagger)(U\otimes V \rho^{SE} U^\dagger\otimes V^\dagger) (U M_i \otimes V) \nonumber \\
    =&(M_iU^\dagger \otimes V^\dagger)(\rho^{SE}_D) (U M_i \otimes V) \nonumber \\
    =& (U^\dagger \otimes V^\dagger)(M_{U,i} \otimes \mathbb{I}) \rho^{SE}_D (M_{U,i} \otimes \mathbb{I})(U \otimes V),
\end{align}
where $\rho_D^{SE}$ has a diagonal correlation matrix, we used the fact that local unitaries do not destroy correlations, and $M_{U,i} = U M_i U^\dagger$ is just another weak measurement of the form of Eq.~\eqref{Eq:special_weak_measurement_apdx}.

Thus, we can narrow down our study to measurements on states whose correlation matrix is diagonal; that is, we now focus our attention to 
\begin{align}
    (M \otimes \mathbb{I}) \rho^{SE}_D (M \otimes \mathbb{I}).
\end{align}
Let us focus on the $M_+$ outcome. The correlation matrix and resultant Bloch vectors for this state are given by
\begin{align}
    \va{b}' &= \dfrac{\va{b} + \epsilon \; \mathcal{T}_D \vu{n}}{1 + \epsilon \; \vu{n}^T \va{a}} \\ \nonumber
    \va{a}' &= \dfrac{\sqrt{1-\epsilon^2} \; \va{a} + (1-\sqrt{1-\epsilon^2}) (\vu{n} \vu{n}^T) \va{a} + \epsilon \; \vu{n}}{1 + \epsilon \; \vu{n}^T \va{a}} \\
    \mathcal{T}_D' &= \dfrac{\sqrt{1-\epsilon^2} \mathcal{T}_D + (1-\sqrt{1-\epsilon^2}) (\vu{n} \vu{n}^T) \mathcal{T}_D + \epsilon \; \vu{n} \va{b}^T}{1 + \epsilon \; \vu{n}^T \va{a}}
\end{align}
for the initial diagonal correlation matrix $\mathcal{T}_D$, system Bloch vector $\vec{a}$, and environment Bloch vector $\vec{b}$. 
Then, we have
\begin{align}
    \va{b}'^T &= \dfrac{\va{b}^T + \epsilon \; \vu{n}^T\mathcal{T}_D}{1 + \epsilon \; \vu{n}^T \va{a}}.
\end{align}
Then
\begin{widetext}
\begin{align}
    \va{a}'\va{b}'^T=\dfrac{\sqrt{1-\epsilon^2}( \va{a}\va{b}^T+\epsilon \; \va{a}\vu{n}^T\mathcal{T}_D) + (1-\sqrt{1-\epsilon^2}) (\vu{n} \vu{n}^T) \va{a}(\va{b}^T + \epsilon \; \vu{n}^T\mathcal{T}_D) + \epsilon \; \vu{n}(\va{b}^T + \epsilon \; \vu{n}^T\mathcal{T}_D)}{(1 + \epsilon \; \vu{n}^T \va{a})^2}
\end{align}
\end{widetext}
We prove one direction by setting $\epsilon = 1$ and showing the equality. Note that $\vec{a}' = \hat{n}$ so that the outer product
\begin{align}
    \vec{a}' \vec{b}'^T = \dfrac{\hat{n} \va{b}^T + \hat{n} \hat{n}^T \mathcal{T}_D}{1 + \; \vu{n}^T \va{a}}.
\end{align}
Plugging in 1 for $\epsilon$ into our definition of $\mathcal{T}_D'$ gives exactly this matrix. Thus, whenever $\epsilon = 1$, we see that $\mathcal{T}_D' = \vec{a}' \vec{b}'^T$ and hence we a have product state. This is true for \textit{all} density operators. Now we must prove the converse to complete the theorem. 

We see that whenever $\mathcal{T}_D' = \vec{a}' \vec{b}'^T$, then
\begin{align}
    &[\sqrt{1-\epsilon^2}+(1-\sqrt{1-\epsilon^2})\hat{n}\hat{n}^T][\eta \mathcal{T}_D - \vec{a} \vec{b}^T - \epsilon \vec{a} \hat{n}^T \mathcal{T}_D] \nonumber \\
    = &\epsilon \hat{n}[(1-\eta)\vec{b}^T + \epsilon \hat{n}^T \mathcal{T}_D],
\end{align}
where $\eta = 1 + \epsilon \; \hat{n}^T \vec{a}$ is the normalization factor. Then
\begin{align}\label{eq:correlcalcConverse2}
    &[\sqrt{1-\epsilon^2}+(1-\sqrt{1-\epsilon^2})\hat{n}\hat{n}^T][(1 + \epsilon \; \hat{n}^T \vec{a}) \mathcal{T}_D\notag\\ 
    &- \vec{a} \vec{b}^T - \epsilon \vec{a} \hat{n}^T \mathcal{T}_D] \nonumber
    = \epsilon^2 \hat{n}[-(\hat{n}^T \vec{a})\vec{b}^T + \hat{n}^T \mathcal{T}_D],
\end{align}
where we substituted for $\eta$.
If we prove that $\epsilon = 1$ is the only solution for any $\vec{a}$ and $\vec{b}$, then the result follows by implication. Substituting $\chi=\mathcal{T}_D-\vec{a} \vec{b}^T$ into Eq.~\eqref{eq:correlcalcConverse2} we have
\begin{align}
    &[\sqrt{1-\epsilon^2}+(1-\sqrt{1-\epsilon^2})\hat{n}\hat{n}^T][\chi+\epsilon \; \hat{n}^T \vec{a}\mathcal{T}_D - \epsilon \vec{a} \hat{n}^T \mathcal{T}_D] \nonumber \\
    &= \epsilon^2 \hat{n} \hat{n}^T\chi.
\end{align}
Using the fact that $\hat{n}^T\vec{a}$ is a number, we get
\begin{align}\label{eq:correlcalcConverse4}
    &\sqrt{1-\epsilon^2}(\chi+\epsilon \; \hat{n}^T \vec{a}\mathcal{T}_D - \epsilon \vec{a} \hat{n}^T \mathcal{T}_D)+(1-\sqrt{1-\epsilon^2})\hat{n}\hat{n}^T\chi \nonumber \\
    &= \epsilon^2 \hat{n} \hat{n}^T\chi.
\end{align}
We can separate out terms in the matrices by projecting onto $\hat{n}\hat{n}^T.$ Acting on both sides of Eq.~\eqref{eq:correlcalcConverse4} with $\hat{n}\hat{n}^T,$ we get
\begin{align}
    &\sqrt{1-\epsilon^2}\hat{n}\hat{n}^T\chi+(1-\sqrt{1-\epsilon^2})\hat{n}\hat{n}^T\chi =\epsilon^2 \hat{n} \hat{n}^T\chi\\
    &\rightarrow\notag\\
    &\label{eq:projectedOntoN}(1-\epsilon^2) \hat{n} \hat{n}^T \chi=\overleftrightarrow{0},
\end{align}
where $\overleftrightarrow{0}$ is the zero matrix. 

Eq.~\eqref{eq:projectedOntoN} can be satisfied when (i): $\epsilon^2=1$, (ii): $\hat{n}=\vec{0}$, (iii): $\hat{n}^T\chi=\vec{0}^T$, or (iv): $\chi=\overleftrightarrow{0}$, where $\vec{0}$ is a column vector of all zeros. We can discard (ii) because it simplifies to (i) or (iv) due to Eq.~\eqref{eq:correlcalcConverse4}. We can also discard (iv) since it implies that the initial state is uncorrelated. Condition (i) is the case we want to prove so we only need to examine (iii). 

Substituting (iii): $\hat{n}^T\chi=\vec{0}$ into Eq.~\eqref{eq:correlcalcConverse4}, we get
\begin{align}\label{eq:orthoNandX1}
    \sqrt{1-\epsilon^2}(\chi+\epsilon\hat{n}^T\vec{a}\mathcal{T}_D- \epsilon \vec{a} \hat{n}^T \mathcal{T}_D)=\overleftrightarrow{0}.
\end{align}
Substituting for $\mathcal{T}_D=\chi+\vec{a}\vec{b}^T$ into Eq.~\eqref{eq:orthoNandX1}, we get
\begin{align}
    &\sqrt{1-\epsilon^2}[\chi+\epsilon\hat{n}^T\vec{a}(\chi+\vec{a}\vec{b}^T)- \epsilon \vec{a} \hat{n}^T (\chi+\vec{a}\vec{b}^T)]=\overleftrightarrow{0}\\
    &\rightarrow\notag\\
    &\label{eq:orthoNandX3}\sqrt{1-\epsilon^2}[\chi+\epsilon\hat{n}^T\vec{a}\chi- \epsilon \vec{a} \hat{n}^T \chi]=\overleftrightarrow{0},
\end{align}
where on the last equation we used the fact that $\hat{n}^T\vec{a}$ is a number. Using again condition (iii), Eq.~\eqref{eq:orthoNandX3} simplifies to
\begin{align}
    &\sqrt{1-\epsilon^2}[(1+\epsilon\hat{n}^T\vec{a})\chi]=\overleftrightarrow{0}\\
    &\rightarrow\notag\\
    &\epsilon\hat{n}^T\vec{a}=-1,
\end{align}
where the last equation comes from the fact that we already considered the other solutions. From the definitions of $\hat{n}$ and $\vec{a}$, the maximum value of $\abs{\hat{n}^T\vec{a}}$ is 1. Thus, it is necessary that $\epsilon=\pm 1$ to reach an uncorrelated output state for an initially correlated input state and the result follows. The results also hold for $M_-$ because, from the definitions in Eq. \eqref{Eq:special_weak_measurement_apdx}, we have that $M_-$ is equal to $M_+$, but with $\hat{n}$ pointing in the opposite direction.
\end{proof}

\section{Proof of Lemma \ref{lemma:weakBrunPure}}
\begin{proof}
Let the two qubit initial pure state  be
\begin{align}\label{eq:initArb2QPureState}
    \ket{\psi}=a\ket{00}+b\ket{01}+c\ket{10}+d\ket{11}.
\end{align} 
The concurrence of the initial state $\ket{\psi}$ is
\begin{align}\label{eq:concAD-BC}
    C(\ket{\psi})=2\abs{ad-bc}.
\end{align}
Since our measurement operators commute, we can think of the problem as $n$ rounds and we choose $m$ of these rounds to be measurement outcome $M_0$. The order that we pick which rounds to be $M_0$ does not matter. Thus, there are $\binom{n}{m}$ possible outcomes for $m$ rounds of $M_0$ measurement outcomes and the mixed post measurement state is
\begin{align}\label{eq:postStatePsi}
    \psi_n=\sum_{m=0}^{n}\binom{n}{m}(M_0^m M_1^{n-m}\otimes\mathbb{I})\op{\psi}(M_0^m M_1^{n-m}\otimes\mathbb{I}).
\end{align}
The concurrence of a two qubit mixed state $\rho$ is given by $\text{max}\{0,\lambda_1-\lambda_2-\lambda_3-\lambda_4\}$, where the $\lambda_i$s are the square roots of the eigenvalues in decreasing order of the matrix $\rho(\sigma_y\otimes\sigma_y)\rho^*(\sigma_y\otimes\sigma_y)$ \cite{Wooters_1998EntOfFormationOfAnArbitraryStateOfTwoQubits}. The concurrence of Eq.~\eqref{eq:postStatePsi} is
\begin{align}\label{eq:concPostMixed}
   C(\psi_n)= (1-\epsilon^2)^{n/2} \cdot C(\psi).
\end{align}

The only time this value is equal to 0 is when a strong projective measurement ($\epsilon = 1$) is applied to the first qubit or when $n \rightarrow \infty$ for $\epsilon < 1$.

\end{proof}

\section{Proof of Theorem \ref{thm:specialWeakMeasEntOneSys}}
\begin{proof}
Let the initial two qubit pure state be
\begin{align}
    \ket{\psi}=a\ket{00}+b\ket{01}+c\ket{10}+d\ket{11}.
\end{align} 
Notice that the operators $M_+$ and $M_-$ commute. 
Using a local unitary $V$, we can rotate the vector $\hat{n}$ such that $n_z=1$. This puts $M_+$ to $A_0$ and $M_-$ to $A_1$. 
Since local unitaries do not affect entanglement, consequently from Lemma 2, we get the same entanglement formula
\begin{align}\label{Eq:mps_concurrence}
   C(\psi_n)= (1-\epsilon^2)^{n/2} \cdot C(\psi). 
\end{align}


These values are nonzero for initial nonzero concurrence, $0\leq\epsilon<1$, and finite $n$.
\end{proof}


\section{Proof of Corollary \ref{coro:specialWeakUknownBothSys}}
\begin{proof}
    Let the initial two qubit pure state be
    \begin{align}
        \ket{\psi}=a\ket{00}+b\ket{01}+c\ket{10}+d\ket{11}.
    \end{align} 
    
    For the case where we do not know the measurement outcomes, it is easy to calculate the mixed state concurrence in terms of the initial entanglement of the pure state when we use the Schmidt form of $\ket{\psi}$. We use superscripts on the measurement operators to allow for different measurements on the first and second qubits. The final state will have the form
    \begin{widetext}
    \begin{align}\label{eq:measureBothSys1}
        \psi'=&\sum_{m,x=0}^{m=n,x=y}\binom{n}{m}\binom{y}{x}\left[(M_+^{(1)})^m(M_-^{(1)})^{n-m}\otimes (M_+^{(2)})^x(M_-^{(2)})^{y-x}\right]\op{\psi}\left[(M_+^{(1)})^m(M_-^{(1)})^{n-m}\otimes(M_+^{(2)})^x(M_-^{(2)})^{y-x}\right]\\
        &\notag\rightarrow\\
        \notag\psi'=&W\otimes V\sum_{m,x=0}^{m=n,x=y}\binom{n}{m}\binom{y}{x}\\
        &\left[(A_0^{(1)})^m(A_1^{(1)})^{n-m}\otimes (A_0^{(2)})^x(A_1^{(2)})^{y-x}\right]\op{\psi}\left[(A_0^{(1)})^m(A_1^{(1)})^{n-m}\otimes(A_0^{(2)})^x(A_1^{(2)})^{y-x}\right]W^\dagger\otimes V^\dagger,
    \end{align}
    \end{widetext}
    where $W$ and $V$ are local unitaries and we use the fact that local unitaries do not change the entanglement. $\psi'$ has a final concurrence of
    \begin{align}
        C(\psi') = (1-\epsilon_1^2)^{n/2} (1-\epsilon_2^2)^{y/2} C(\psi).
    \end{align}
    Note that the entanglement of $\psi'$ only goes to zero when perfect projective measurements are applied by either party or when either party repeats their measurements an infinite amount of times for non-projective measures. The values $\epsilon_1$ and $\epsilon_2$ represent the strength of the measurements applied to systems 1 and 2, respectively.
\end{proof}

\section{Interesting Example of Destroying Entanglement with Unknown Measurement Outcomes}

Define our local measurements to be
\begin{align}
    M_{\pm} = \sqrt{\dfrac{1\mp\epsilon}{2}} \left( \mathbb{I} \pm \epsilon(\mathbb{I}+\hat{\sigma}_x) \right)
\end{align} 
and the local pure state to have the Schmidt form of
\begin{align}
    \ket{\psi} = \text{cos}(\theta/2)\ket{00} + \text{sin}(\theta/2)\ket{11},
\end{align}
where $0\leq \epsilon\leq1$ and $0\leq\theta\leq \pi$.

Notice that the final concurrence, after a measurement with an unknown outcome, is given by 0 whenever $\epsilon = 1/\sqrt{2}$. The measurements have a determinant of $\pm 1/\sqrt{2}$ and are not projective. 
The final state has the form
\begin{align}
    \dfrac{1}{4}
    \left(
    \begin{array}{cccc}
        1+\text{cos}(\theta) & 0 & 0 & \text{sin}(\theta) \\
        0 & 1-\text{cos}(\theta) & \text{sin}(\theta) & 0 \\
        0 & \text{sin}(\theta) & 1+\text{cos}(\theta) & 0 \\
        \text{sin}(\theta) & 0 & 0 & 1-\text{cos}(\theta)
    \end{array}
    \right), \nonumber
\end{align}
which has a partial transpose with a determinant of 0.

It turns out that the entanglement is broken for \textit{any} pure state only if $\epsilon = 1/\sqrt{2}$. Note that each possible measurement outcome is entangled and the final state is one of these outcomes. The entanglement that is broken is the entanglement of the final state that we can \textit{assign}, which is limited by our ignorance.

%

\vfill

\small
\framebox{\parbox{\linewidth}{
The submitted manuscript has been created by UChicago Argonne, LLC, Operator of 
Argonne National Laboratory (``Argonne''). Argonne, a U.S.\ Department of 
Energy Office of Science laboratory, is operated under Contract No.\ 
DE-AC02-06CH11357. 
The U.S.\ Government retains for itself, and others acting on its behalf, a 
paid-up nonexclusive, irrevocable worldwide license in said article to 
reproduce, prepare derivative works, distribute copies to the public, and 
perform publicly and display publicly, by or on behalf of the Government.  The 
Department of Energy will provide public access to these results of federally 
sponsored research in accordance with the DOE Public Access Plan. 
http://energy.gov/downloads/doe-public-access-plan.}}

\end{document}